\RequirePackage{amsmath}
\documentclass[a4paper, 10pt]{amsart}
\usepackage{amssymb, url, color, pb-diagram, graphicx, amscd, pb-diagram,mathrsfs}
\usepackage[colorlinks=true, bookmarks=true, pdfstartview=FitH, pagebackref=true]{hyperref}
\usepackage[nodayofweek]{datetime}
\usepackage{enumerate,stmaryrd} 
\usepackage{euscript}
\usepackage{verbatim}
\usepackage{graphicx}
\usepackage{times}
\usepackage{arydshln} 
\usepackage{setspace}
\usepackage[space, compress, sort]{cite}
\usepackage{empheq}

\numberwithin{equation}{section}

\theoremstyle{plain}

\newtheorem{thm}{Theorem}[section]
\newtheorem{theorem}[thm]{Theorem}

\newtheorem{claim}{Claim}

\theoremstyle{remark}

\newtheorem{remark}[thm]{Remark}

\theoremstyle{definition}

\newcounter{mnotecount}[section]

\renewcommand{\phi}{\varphi}

\newcommand{\bR}{\mathbb{R}}
\newcommand{\bL}{\mathbb{L}}

\newcommand{\cC}{\mathcal{C}}
\newcommand{\cM}{\mathcal{M}}

\newcommand{\cY}{\mathcal{Y}}

\newcommand{\ghat}{{\widehat{g}}}
\newcommand{\Khat}{\widehat{K}}

\newcommand{\nuhat}{\widehat{\nu}}
\newcommand{\nhat}{\widehat{n}}

\newcommand{\nablahat}{\widehat{\nabla}}

\newcommand{\gbar}{\overline{g}}

\renewcommand{\hbar}{\overline{h}}
\newcommand{\nubar}{\overline{\nu}}

\newcommand{\phibar}{\overline{\varphi}}
\newcommand{\Wtil}{\widetilde{W}}
\newcommand{\Ztil}{\widetilde{Z}}

\newcommand{\thetatil}{\widetilde{\theta}}

\newcommand{\phitil}{\widetilde{\phi}}
\newcommand{\xitil}{\widetilde{\xi}}
\newcommand{\sigmatil}{\widetilde{\sigma}}

\newcommand{\vtil}{\widetilde{v}}

\newcommand{\ubar}{\overline{u}}

\newcommand{\definedas}{\mathrel{\raise.095ex\hbox{\rm :}\mkern-5.2mu=}}

\let\<\langle 
\let\>\rangle

\DeclareMathOperator{\tr}{tr}
\DeclareMathOperator{\divg}{div}

\newcommand{\scal}{\mathrm{Scal}}

\DeclareMathOperator{\DeltaL}{\Delta_{\bL, \mathit g}}
\newcommand{\gbold}{\boldsymbol g}

\newcommand{\dv}{\, d{\rm vol}}
\newcommand{\ds}{\, d{\rm s}}

\begin{document}

\title[Small TT-tensor solutions to the constraint equations]
{A new point of view on the solutions to the Einstein constraint
equations with arbitrary mean curvature and small TT-tensor}

\begin{abstract}
In this short note, we give a construction of solutions to the Einstein constraint equations
using the well known conformal method. Our method gives a result similar to the one in
\cite{HNT1,HNT2,MaxwellNonCMC}, namely existence when the so called TT-tensor $\sigma$ is small
and the Yamabe invariant of the manifold is positive. The method we describe is however much
simpler than the original method and allows easy extensions to several other problems. Some
non-existence results are also considered.
\end{abstract}

\author[R. Gicquaud]{Romain Gicquaud}
\address[R. Gicquaud]{Laboratoire de Math\'ematiques et de Physique Th\'eorique \\ Universit\'e Fran\c cois Rabelais de Tours \\ Parc de Grandmont\\ 37200 Tours \\ FRANCE}
\email{\href{mailto: R. Gicquaud <Romain.Gicquaud@lmpt.univ-tours.fr>}{romain.gicquaud@lmpt.univ-tours.fr}}

\author[Q.A. Ng\^{o}]{Qu\^{o}\hspace{-0.5ex}\llap{\raise 1ex\hbox{\'{}}}\hspace{0.5ex}c Anh Ng\^{o}}
\address[Q.A. Ng\^{o}]{Laboratoire de Math\'ematiques et de Physique Th\'eorique \\ Universit\'e Fran\c cois Rabelais de Tours \\ Parc de Grandmont\\ 37200 Tours \\ FRANCE}
\email{\href{mailto: Q. A. Ngo <Quoc-Anh.Ngo@lmpt.univ-tours.fr>}{quoc-anh.ngo@lmpt.univ-tours.fr}}
\email{\href{mailto: Q. A. Ngo <bookworm\_vn@yahoo.com>}{bookworm\_vn@yahoo.com}}

\date{\today \ at \currenttime}


\maketitle

\tableofcontents

\section{Introduction}\label{secIntro}

\subsection{The Einstein constraint equations}\label{secConstraints}

Initial data for the Cauchy problem in general relativity are usually given in terms of the
geometry of the Cauchy surface $(M, \ghat)$ in the spacetime $(\cM, \gbold)$ of dimension $n+1$
with $n \geqslant 3$. Assuming that the spacetime $\cM$ is globally hyperbolic and $M$ is a
spacelike Cauchy surface, one can define the metric $\ghat$ induced on $M$ by the spacetime metric
$\gbold$ and the second fundamental form $\Khat$ of $M$ in $\cM$. It follows from the Einstein
equations together with the Gauss and Codazzi equations that $\ghat$ and $\Khat$ are related by
the following equations
\begin{subequations}\label{eqConstraintEquations}
\begin{empheq}[left=\empheqlbrace]{align}
\label{eqHamiltonian}
\scal_\ghat + (\tr_\ghat \Khat )^2 - |\Khat|_\ghat^2 & =2 \rho,\\
\label{eqMomentum}
\divg_\ghat \Khat - d (\tr_\ghat \Khat ) & = j,
\end{empheq}
\end{subequations}
where $\rho$ and $j$ are related to the other fields such as matter fields, electromagnetic field, etc.
that one wants to include into the universe under consideration. Also in \eqref{eqConstraintEquations}, $\scal_\ghat$ is the scalar
curvature of $\ghat$. To keep things simple, we will consider no field but the gravitational field,
hence, forcing $\rho \equiv 0$ and $j \equiv 0$.

A simple dimension counting argument shows that the system \eqref{eqConstraintEquations} is
under-determined, thus, it is generally hard to solve \eqref{eqConstraintEquations} in this form. To overcome this difficulty, we need
to decompose both $\ghat$ and $\Khat$ into given data and unknowns that will have to be adjusted so that
Equations \eqref{eqHamiltonian} and \eqref{eqMomentum} are fulfilled. Several such splitting exist and
we refer the reader to \cite{BartnikIsenberg} for a detailed review of some known results on the constraint
equations. In the literature, the most commonly used method is the conformal method which we briefly describe now.
We invite the reader to have a look at the very nice recent work of Maxwell
\cite{MaxwellConformalParameterization,MaxwellConformalMethod,MaxwellInitialData} for a deep understanding of this
method and its connection to other widely used methods.

The given data in the conformal method consist of
\begin{itemize}
 \item a Riemannian manifold $(M, g)$,
 \item a function $\tau: M \to \bR$,
 \item and a symmetric $2$--tensor $\sigma$ on $M$ which is traceless and transverse in the following sense
\[\tr_g \sigma \equiv 0, \quad \divg_g \sigma \equiv 0.\]
\end{itemize}

\noindent As a shorthand, we will call $\sigma$ a TT-tensor. The unknowns in the conformal method are
\begin{itemize}
 \item a positive function $\phi: M \to \bR^*_+$,
 \item and a $1$--form $W$.
\end{itemize}

Combining all these elements, one can form $(\ghat, \Khat)$ as follows:
\begin{equation}\label{eqDecomposition}
\begin{aligned}
\ghat & = \phi^{N-2} g,\\
\Khat & = \frac{\tau}{n} \ghat + \phi^{-2} \left(\sigma + \bL_g W\right),
\end{aligned}
\end{equation}
where $N \definedas 2n/(n-2)$ and $\bL_g$ is the conformal Killing operator given by
\begin{equation*}\label{eqConformalKillingOperator}
\bL_g W_{ij} \definedas \nabla_i W_j + \nabla_j W_i - \frac{2}{n} \nabla^k W_k g_{ij},
\end{equation*}
with $\nabla$ the Levi--Civita connection associated to the background metric $g$. Here, $\tau$
is the mean curvature of $M$ as an hypersurface of $(\cM, \gbold)$ given by 
\[\tau = \ghat^{ij} \Khat_{ij}.\]
The choice for $\sigma$ and $W$ in \eqref{eqDecomposition} is related to the York splitting,
see the remark at the end of Section \ref{secBoundaryConditions}.

Using the decomposition \eqref{eqDecomposition}, the constraint equations \eqref{eqConstraintEquations} become
a system of PDEs for $(\phi, W)$ given as follows:
\begin{subequations}\label{systemWithoutB}
\begin{empheq}[left=\empheqlbrace]{align}
-\frac{4(n-1)}{n-2} \Delta_g \phi + \scal_g \phi &= - \frac{n-1}{n} \tau^2 \phi^{N-1} + \big| \sigma + \bL_g W \big|_g^2 \phi^{-N-1}, \label{eqLichnerowicz}\\
\DeltaL W &= \frac{n-1}{n} \phi^N d\tau, \label{eqVector}
\end{empheq}
\end{subequations}
where we denote $\Delta_g \phi = \divg_g(\nabla_g \phi )$ and $\DeltaL W = \divg_g (\bL_g W)$. In the literature, Equation
\eqref{eqLichnerowicz} is commonly known as the Lichnerowicz equation while Equation \eqref{eqVector} is usually
called the vector equation.

The system \eqref{systemWithoutB} is notoriously hard to solve except in the case when $\tau$ is a constant
function which is now well understood, see for instance \cite{Isenberg}. Indeed, when $\tau$ is constant, Equation \eqref{eqVector}
only involves $W$ and generically implies that $W \equiv 0$. Therefore, one is left with solving the Lichnerowicz
equation \eqref{eqLichnerowicz} without any $W$. However, everything dramatically changes when $\tau$ is no longer
constant. Perturbation arguments can be used to address the case when $d\tau$ is small in some sense. But, until
recently, very few results were known for arbitrary $\tau$. Two major breakthroughs were obtained first by M. Holst,
G. Tsosgtgerel, and G. Nagy in \cite{HNT1,HNT2}, by D. Maxwell in \cite{MaxwellNonCMC}, and then by M. Dahl, E. Humbert,
and the first author in \cite{DahlGicquaudHumbert}.

Usually, standard methods to solve elliptic PDEs require an a priori knowledge of the solutions, i.e. nice domains in which one
can try to apply fixed point theorems, fixed point arguments, etc. However, via a simple scaling argument, changing
$\phi$ to $\lambda \phi$ where $\lambda \gg 1$ shows that the two dominant terms in the Lichnerowicz equation are
$\frac{n-1}{n} \tau^2 \phi^{N-1}$ and $| \bL_g W|_g^2 \phi^{-N-1}$. These two terms have the same scaling behavior
but come up with opposite signs in the Lichnerowicz equation \eqref{eqLichnerowicz}. Although the first term has the
right sign and in fact helps us in applying the maximum principle, the second one has the wrong sign and eventually
destroys any attempt to get an a priori upper bound for $\phi$ when $d\tau$ is not small.

\subsection{The Holst--Nagy--Tsogtgerel--Maxwell method}\label{secHNT}

Losing such an a priori estimate, a very nice idea was proposed in \cite{HNT1,HNT2}. The idea was pushed further
in \cite{MaxwellNonCMC}. It consists in looking for
solutions of the system \eqref{systemWithoutB} with $\phi$ and $W$ very close to zero to make the two dominant
terms irrelevant. To do this, they require the manifold $(M, g)$ to be closed with a positive Yamabe invariant,
$\cY(g)>0$ (see Equation \eqref{eqDefYamabe} below). Consequently, the scalar curvature $\scal_g$ becomes in some
sense the dominant term. In addition, they also require that $\sigma$ is small to control the right hand side
of Equation \eqref{eqLichnerowicz}. The theorem they obtained is the following.

\begin{theorem}[see \cite{MaxwellNonCMC}]\label{thmMaxwell}
Let $M$ be a compact Riemannian manifold without boundary. Given $p > n$, let $g \in W^{2, p}$, $\tau \in W^{1, p}$, and
$\sigma \in W^{1, p}$, $\sigma \not\equiv 0$ be given data. Assume that the Yamabe invariant $\cY(g)$ is strictly
positive and that $g$ has no conformal Killing vector fields. Then, if $\left\|\sigma\right\|_{L^\infty}$ is small
enough, there exists at least one solution $(\phi, W) \in W^{2, p}(M, \bR) \times W^{2, p}(M, T^*M)$ to the system
\eqref{systemWithoutB}.
\end{theorem}

Assume that $M$ is a compact manifold without boundary, we recall that the Yamabe invariant $\cY(g)$ of $(M, g)$ is defined as
\begin{equation}\label{eqDefYamabe}
\cY(g)=\inf_{0 \not\equiv \phi \in W^{1, 2}(M, \bR) } \frac{\int_M \big(  \frac{4(n-1)}{n-2} |d\phi|^2_g + \scal_g \phi^2  \big) \dv_g}{\left(\int_M \phi^N \dv_g\right)^{N/2}}
\end{equation}

The method of \cite{HNT1,HNT2} was recently adapted to other situation such as asymptotically Euclidean manifolds in
\cite{DiltsIsenbergMazzeoMeier}, asymptotically cylindrical manifolds in \cite{Leach}, compact manifolds with boundary
in \cite{Dilts,HolstMeierTsogtgerel}, and to asymptotically Euclidean manifolds with boundary in \cite{HolstMeier}. As
can be seen from the statement of Theorem \ref{thmMaxwell} and as we have just mentioned earlier, the smallness of
$\|\sigma\|_{L^\infty}$ was used. However, it is worth mentioning that such an $L^\infty$--smallness assumption can be
weaken to $\left\|\sigma\right\|_{L^2}$ small enough, see \cite{Nguyen}.

\subsection{The Dahl--Gicquaud--Humbert method}\label{secDGH}

The idea of \cite{DahlGicquaudHumbert} goes in the opposite direction to the method in Subsection \ref{secHNT}.
Intuitively, the idea of \cite{DahlGicquaudHumbert} is to study what happens if $\phi$ and $W$ can become very
large, i.e. what prevents the existence of an a priori estimate. The answer to this question is heuristically
that if $\phi$ can become very large, by setting $\gamma = \left\|\phi\right\|_{L^\infty}$ and by
renormalizing $\phi$, $W$, and $\sigma$ as follows:
\[
\phitil = \gamma^{-1} \phi, \quad \Wtil =\gamma^{-N} W, \quad \sigmatil = \gamma^{-N} \sigma,
\]
it turns out that $\phitil$ and $\Wtil$ satisfy the following system
\[
\left\lbrace
\begin{aligned}
\frac{1}{\gamma^{N-2}} \left( -\frac{4(n-1)}{n-2} \Delta_g \phitil + \scal_g \phitil \right) & =- \frac{n-1}{n} \tau^2 \phitil^{N-1}+
 \big|\sigmatil + \bL_g\Wtil\big|_g^2 \phitil^{-N-1},\\
\DeltaL\Wtil & = \frac{n-1}{n} \phitil^N d\tau.
\end{aligned}
\right.
\]
In the limit as $\gamma \to +\infty$, one is left with
\[
\frac{n-1}{n} \tau^2 \phitil^{N-1} = \big|\bL_g \Wtil\big|_g^2 \phitil^{-N-1}.
\]
Therefore, $\Wtil$ becomes a non-trivial solution to the so-called limit equation
\begin{equation}\label{eqLimit0}
\DeltaL \Wtil = \sqrt{\frac{n-1}{n}} \big|\bL_g\Wtil\big|_g \frac{d\tau}{\tau}.
\end{equation}
The rigorous argument leads to a similar limit equation with a parameter $\alpha \in (0, 1]$ given as follows:
\begin{equation}
\label{eqLimit}
\DeltaL\Wtil = \alpha \sqrt{\frac{n-1}{n}} \big|\bL_g\Wtil\big|_g \frac{d\tau}{\tau}.
\end{equation}
The main theorem of \cite{DahlGicquaudHumbert} can be stated as follows.

\begin{theorem}\label{thmLimit}
Let $M$ be a compact Riemannian manifold without boundary. Given $p > n$, let $g \in W^{2, p}$, $\tau \in W^{1, p}$, and
$\sigma \in W^{1, p}$ be given data. Assume that $g$ has no conformal Killing vector fields, $\tau > 0$ and that
$\sigma \not\equiv 0$ if $\cY(g) \geqslant 0$. If the limit equation \eqref{eqLimit} admits no non-zero solution
$\Wtil$ for all values of the parameter $\alpha \in (0, 1]$, then there exists at least one solution
$(\phi, W) \in W^{2, p}(M, \bR) \times W^{2, p}(M, T^*M)$ to the system \eqref{systemWithoutB}.
\end{theorem}

It is worth noticing that the result in \cite{DahlGicquaudHumbert} requires that $\tau$ is bounded away
from zero, however, it involves no assumption on the Yamabe invariant $\cY(g)$.
A simplified proof of Theorem \ref{thmLimit} appears in \cite{Nguyen}.

This method was adapted to several other contexts such as asymptotically hyperbolic manifolds in
\cite{GicquaudSakovich} and asymptotically cylindrical manifolds in \cite{DiltsLeach}. In particular,
strong results are obtained for negatively curved manifolds, see
\cite[Proposition 6.2 and Remark 6.3]{GicquaudSakovich}. The case of asymptotically Euclidean manifolds
and compact manifolds with boundary are currently work in progress \cite{DiltsGicquaudIsenberg,GicquaudHumbertNgo}.
New difficulties show up in these cases.

\subsection{Objective and outline of the paper}

As we have already seen from Subsections \ref{secHNT} and \ref{secDGH}, both approaches we presented are
dual in a certain sense. The first one constructs solutions which are very close to zero while the second
one is a means to ensure control on the size of the solutions. In this note, we emphasize the duality between
both methods showing that the Holst et al. method can be rephrased as a scaling argument. This duality
can potentially be deepen further, recasting both methods in a single framework, see Remark \ref{rkLimit}.
This also sheds a light on the role of the assumptions of the main theorem of \cite{MaxwellNonCMC}.

Nevertheless, our new method leads to a result which is not as good as the one of \cite{MaxwellNonCMC} but
it is much simpler than the original one and appears also quite versatile.

In Section \ref{secClosed}, we present in detail the simplest case of our method, namely when the manifold
is closed. Also in this section, a non-existence result is presented. Then, we give a quick look at the
asymptotically Euclidean case in Section \ref{secAE} and at compact manifolds with boundary in Section
\ref{secCWB}.

\section{The closed case}\label{secClosed}

In this section, we are interested in studying solutions of \eqref{systemWithoutB}
when the underlying manifold $M$ is compact without boundary. In the first part of this section, we prove a
result which basically says that \eqref{systemWithoutB} is solvable when $\cY(g)>0$ and $\sigma \not \equiv 0$
is small enough, see Theorem \ref{thmClosed} below. Then, we improve \cite[Theorem 1.7]{DahlGicquaudHumbert}
by showing that \eqref{systemWithoutB} admits no solution provided $\cY(g)>0$, $\sigma \equiv 0$, and $d\tau /\tau$
is small in the $L^n$--norm.

\subsection{Existence results for small but non-vanishing TT-tensor}

The main result of this subsection is the following.

\begin{theorem}\label{thmClosed} 
Let $M$ be a compact manifold without boundary. Given $p > n$, let $g \in W^{2, p}(M, S^2(M))$,
$\tau \in W^{1, p}(M, \bR)$ and $\sigmatil \in W^{1, p}(M, S^2(M))$, $\sigmatil \not\equiv 0$ be given data. Assume that
the Yamabe invariant $\cY(g)$ is strictly positive and that $g$ has no conformal Killing vector fields. There exists
$\eta_0 > 0$ such that for any $\eta \in (0, \eta_0)$ there exists at least one solution
$(\phi, W) \in W^{2, p}(M, \bR) \times W^{2, p}(M, T^*M)$ to the system \eqref{systemWithoutB} with $\sigma = \eta \sigmatil$. 
\end{theorem}

Note that this theorem is not as good as Theorem \ref{thmMaxwell}. Indeed, $\eta_0$ depends a priori on $\sigmatil$
in an unknown way while Theorem \ref{thmMaxwell} asserts that the system \eqref{systemWithoutB} with $\sigma = \eta \sigmatil$
has a solution provided that $\left\|\sigma\right\|_{L^\infty} = \left|\eta\right| \left\|\sigmatil\right\|_{L^\infty}$
is small enough (less than some $\epsilon > 0$). So the corresponding $\eta_0$ would be
$\epsilon / \left\|\sigmatil\right\|_{L^\infty}$. Nevertheless, the proof appears to be constructive since
it relies on the sub- and super-solutions method and on the implicit function theorem. For the sake of clarity, we divide
the proof into several claims.

\begin{claim}\label{clLichneroLimit}
Let $\sigmatil \not\equiv 0$ be a TT-tensor belonging to $W^{1, p}(M, S^2(M))$. Then there exists a unique solution $\phitil_0 \in W^{2, p}(M, \bR)$
to the following equation
\begin{equation}\label{eqLichnerowicz0}
-\frac{4(n-1)}{n-2} \Delta_g \phitil + \scal_g \phitil = \left|\sigmatil\right|_g^2 \phitil^{-N-1},
\end{equation}
\end{claim}

\begin{proof}
The proof is standard, see \cite{MaxwellRoughCompact}. Note that this equation is nothing but the Lichnerowicz equation with
$\tau \equiv 0$ and $W \equiv 0$. To prove existence, we rely on the classical sub- and super-solutions method described,
for example, in \cite[Proposition 2]{Isenberg}. Since $\cY(g) > 0$, there exists a positive $W^{2,p}(M)$ function $\psi$ so
that the metric $\gbar = \psi^{N-2} g$ has positive constant scalar curvature. Setting $\phibar = \psi^{-1} \phitil$,
Equation \eqref{eqLichnerowicz0} transforms into
\begin{equation}\label{eqLichConf}
-\frac{4(n-1)}{n-2} \Delta_{\gbar} \phibar + \scal_{\gbar} \phibar = \big|\psi^{-2} \sigmatil\big|_{\gbar}^2 \phibar^{-N-1}.
\end{equation}
To solve \eqref{eqLichConf} for $\phibar$, we follows the method of sub- and super-solutions by constructing a sub-solution
$\phibar_-$ and a super-solution $\phibar^+$ as follows. Let $\ubar \in W^{2, p}(M)$ denote the solution to the following linear equation:
\[
-\frac{4(n-1)}{n-2} \Delta_{\gbar} \ubar + \scal_{\gbar} \ubar = \big| \psi^{-2} \sigmatil\big|_{\gbar} ^2.
\]
It follows from the strong maximum principle that $\ubar > 0$ in $M$. By setting
\[
\phibar_- = \left(\max \ubar\right)^{-\frac{N+1}{N+2}} \ubar
\]
and
\[
\phibar_+ = \left(\min \ubar\right)^{-\frac{N+1}{N+2}} \ubar,
\]
one readily checks that $\phibar_+$ and $\phibar_-$ are super- and sub-solutions for \eqref{eqLichConf} respectively, meaning that
\[
-\frac{4(n-1)}{n-2} \Delta_{\gbar} \phibar_- + \scal_{\gbar} \phibar_- \leqslant \left|\psi^{-2} \sigmatil\right|_{\gbar}^2 (\phibar_-)^{-N-1}
\]
and that
\[
-\frac{4(n-1)}{n-2} \Delta_{\gbar} \phibar_+ + \scal_{\gbar} \phibar_+ \geqslant \left|\psi^{-2} \sigmatil\right|_{\gbar}^2 (\phibar_+)^{-N-1}.
\]
Hence, there exists (at least) one solution $\phibar$ to Equation \eqref{eqLichConf} and it leads to a solution
$\phitil_0 = \psi \phibar$ to Equation \eqref{eqLichnerowicz0} as well. 

Uniqueness is also easy to prove. Indeed, let $\phitil_0$ and $\phitil'_0$ be two solutions to Equations
\eqref{eqLichnerowicz0} and denote $\phibar_0 = \psi^{-1} \phitil_0$ and $\phibar'_0 = \psi^{-1} \phitil'_0$.
A simple calculation leads us to the following equality:
\[\begin{aligned}
-\frac{4(n-1)}{n-2} \Delta_{\gbar} \big(\phibar_0 -\phibar'_0\big) & + \scal_{\gbar} \left(\phibar_0 - \phibar'_0\right)\\
 = & \left|\psi^{-2} \sigmatil\right|_{\gbar}^2 \left( \frac{1}{(\phibar_0)^{N+1}} - \frac{1}{(\phibar'_0)^{N+1}}\right)\\
 = &- \underbrace{(N+1) \left|\psi^{-2} \sigmatil\right|_{\gbar}^2 \int_0^1 \frac{dx}{\left(x \phibar_0 + (1-x) \phibar'_0\right)^{N+2}}}_{:= f}
  \left(\phibar_0 - \phibar'_0\right),
\end{aligned}\]
where the term $f$ is obviously non-negative. This then implies
\[
-\frac{4(n-1)}{n-2} \Delta_{\gbar} \left(\phibar_0 - \phibar'_0\right) + \left(\scal_{\gbar} + f\right)\left(\phibar_0 - \phibar'_0\right)
 = 0.
\]
Since $\scal_{\gbar} + f > 0$, we immediately conclude that $\phibar_0 - \phibar'_0 \equiv 0$. This proves the uniqueness
of the solution $\phitil_0$ as claimed.
\end{proof}

\begin{remark}
As can be seen, the existence of such a metric $\gbar$ in the proof of Claim \ref{clLichneroLimit} does not need the full
strength of the Yamabe theorem, we could only require that $\gbar$ has positive scalar curvature. However, this claim
strongly relies on the positivity of the Yamabe invariant $\cY(g)$. Indeed, assume that there exists a positive solution
$\phitil$ to Equation \eqref{eqLichnerowicz0}, the scalar curvature $\scal_{\ghat}$ of the metric $\ghat = \phi^{N-1} g$
satisfies
\[
\scal_{\ghat} = \phi^{1-N} \left(-\frac{4(n-1}{n-2} \Delta_g \phi + \scal_g \phi\right)
 = \phi^{-2N} \left|\sigmatil\right|^2_g.
\]
Hence, the scalar curvature of $\ghat$ is non-negative and not identically zero.
Thus, $\cY(g) = \cY(\ghat) > 0$. This partially explains why this method cannot be adapted to asymptotically
hyperbolic manifolds.
\end{remark}

Now, we introduce the following $\mu$--deformed system of \eqref{systemWithoutB}:
\begin{subequations}\label{DeformedsystemWithoutB}
\begin{empheq}[left=\empheqlbrace]{align}
-\frac{4(n-1)}{n-2} \Delta_g \phitil + \scal_g \phitil & =- \frac{n-1}{n} \tau^2 \mu^2 \phitil^{N-1} + \big|\sigmatil + \bL_g\Wtil\big|_g^2 \phitil^{-N-1}, \label{eqDeformedLich}\\
\DeltaL\Wtil & = \frac{n-1}{n} \phitil^N \mu d\tau. \label{eqDeformedVector}
\end{empheq}
\end{subequations}
Note that this system is obtained from \eqref{systemWithoutB} by changing the mean curvature $\tau$ simply by $\mu \tau$.

\begin{claim}\label{clDeformed}
There exists $\varepsilon > 0$ such that the system \eqref{DeformedsystemWithoutB} admits a solution
$(\phitil_\mu, \Wtil_\mu) \in W^{2,p}(M, \bR) \times W^{2,p}(M, T^*M)$ for all $\mu \in [0, \varepsilon)$.
\end{claim}

\begin{proof}
The proof is based on the implicit function theorem. First, we define the operator 
\[
F: \bR \times W^{2, p}_+(M, \bR) \times W^{2, p}(M, T^*M) \to L^p(M, \bR) \times L^p(M, T^*M)
\] 
as follows:
\[{\setstretch{2.00}
F(\mu, \phitil, \Wtil) =
\begin{pmatrix} 
-\frac{4(n-1)}{n-2} \Delta_g \phitil + \scal_g \phitil + \frac{n-1}{n} \tau^2 \mu^2 \phitil^{N-1} - \big|\sigmatil + \bL_g\Wtil\big|_g^2 \phitil^{-N-1}\\
\DeltaL\Wtil - \frac{n-1}{n} \phitil^N \mu d\tau
\end{pmatrix}.}
\]
It is readily checked that $F$ is a $C^1$--mapping. Notice that
\[F(0, \phitil_0, 0) = \begin{pmatrix} 0\\0 \end{pmatrix},\] 
where $\phitil_0$ is the solution found in Claim \ref{clLichneroLimit}. All we need to do is to prove that the
partial derivative of $F$ with respect to $(\phitil, \Wtil)$ is an isomorphism at $(0, \phitil_0, 0)$. To this
end, we first observe that the differential $\mathcal D F_{(0, \phitil_0, 0)}$ is given by
\[\begin{split}
\mathcal D F_{(0, \phitil_0, 0)} &(0, \thetatil, \Ztil) \\
= &{\setstretch{2.00}
\left(
\begin{array}{c:c}
-\frac{4(n-1)}{n-2} \Delta_g + \scal_g + (N+1) \left|\sigmatil\right|_g^2 \phitil_0^{-N-2} & -2 \phitil_0^{-N-1} \big\langle\sigmatil, \bL_g \cdot\big\rangle\\
\hdashline 0 & \DeltaL 
\end{array}
\right)   } 
\begin{pmatrix} 
 \thetatil \\ 
 \Ztil 
\end{pmatrix} .\end{split}
\]
Note that $\mathcal DF_{(0, \phitil_0, 0)}(0, \thetatil, \Ztil)$ is triangular, meaning that the second
line of the $2$-by-$2$ block matrix above does not depend on $\thetatil$. Thus, the invertibility of
$\mathcal DF_{(0, \phitil_0, 0)}$ follows from the fact that the diagonal terms
\[
\begin{array}{rrll}
H: & W^{2, p}(M, \bR) & \to & L^p(M, \bR)\\
 & \thetatil & \mapsto & -\frac{4(n-1)}{n-2} \Delta_g \thetatil + \scal_g \thetatil + (N+1) \left|\sigmatil\right|_g^2 \phitil_0^{-N-2} \thetatil
\end{array}
\]
and
\[
\begin{array}{rrll}
V: & W^{2, p}(M, T^*M) & \to & L^p(M, T^*M)\\
 & \Ztil & \mapsto & \DeltaL\Ztil
\end{array}
\]
are invertible. Invertibility of $V$ follows from \cite[Proposition 5]{MaxwellNonCMC}, while $H$ is a
Fredholm map of index $0$. Since $\cY(g) > 0$, the conformal Laplacian is positive definite. Hence,
for any given $u \in W^{2, p}(M)$ with $u \not\equiv 0$, we calculate to obtain
\[
\begin{aligned}
\int_M u H(u) \dv_g
 = &\underbrace{
 \int_M \left(\frac{4(n-1)}{n-2} \left|du\right|^2_g + \scal_g u^2\right) \dv_g}_{> 0} \\
& + \underbrace{
 \int_M (N+1) \left|\sigmatil\right|_g^2 \phitil^{-N-2} u^2 \dv_g}_{\geqslant 0}
 > 0.
\end{aligned}
\]
Hence, $H$ has a trivial kernel. Thus, we have shown that $\mathcal DF_{(0, \phitil_0, 0)}$ is an
isomorphism as claimed.
\end{proof}

The last claim is just a straightforward calculation, therefore we omit its proof.

\begin{claim}\label{clScaling}
Set
\[
\left\lbrace
\begin{aligned}
\phi_\mu  & = \mu^{\frac{2}{N-2}} \phitil_\mu,\\
W_\mu   & = \mu^{\frac{N+2}{N-2}} \Wtil_\mu,\\
\sigma_\mu & = \mu^{\frac{N+2}{N-2}} \sigmatil.
\end{aligned}
\right.
\]
If $(\phitil_\mu, \Wtil_\mu)$ solves \eqref{DeformedsystemWithoutB}, the $(\phi_\mu, W_\mu)$ solves
\eqref{systemWithoutB} with $\sigma = \sigma_\mu$.
\end{claim}

Finally, the proof of Theorem \ref{thmClosed} follows by setting $\eta_0 = \varepsilon^{\frac{N+2}{N-2}}$,
where $\varepsilon$ is the constant appearing in Claim \ref{clDeformed}.

\begin{remark}\label{rkLimit}
It is quite appealing to use the deformation \eqref{DeformedsystemWithoutB} of the conformal constraint
equations to get a new proof of the limit equation criterion as in \cite{DahlGicquaudHumbert}. Indeed,
the system \eqref{DeformedsystemWithoutB} could be studied using the Leray--Schauder fixed point theorem,
which would allow $\mu$ to go up to $1$ (hence $\sigmatil$ would be set equal to the desired $\sigma$).
Assuming that the set of $(\phitil, \Wtil, \mu)$ solutions to \eqref{DeformedsystemWithoutB} with
$0 \leqslant \mu \leqslant 1$ is bounded, the Leray--Schauder theorem would guarantee that the system
\eqref{systemWithoutB} has (at least) one solution. If this set is unbounded, the argument presented
in Section \ref{secDGH} would lead to the existence of a non-trivial solution to Equation \eqref{eqLimit0}.
Hence, the main result of \cite{DahlGicquaudHumbert} could be strengthened, getting rid of the parameter
$\alpha$ (which appears because we introduce a different deformed system there). Such a result would
show that the methods of \cite{HNT1,HNT2} and \cite{DahlGicquaudHumbert} are two facets of a deeper
method. However, one serious difficulty appears in attempting this proof: one has to ensure that
if $\phitil$ (or $\Wtil$) diverges, $\mu$ stays away from 0.
\end{remark}

\subsection{A non-existence result}

The assumption on $\sigma$, namely that it has to be small but cannot be zero, looks weird at first
sight and one can wonder if the hypothesis $\sigma \not\equiv 0$ is purely technical. As can be seen
from \cite{IsenbergOMurchadha, MaxwellNonCMC, DahlGicquaudHumbert}, $\sigma$ is used to show that the
function $\phi$ solving the Lichnerowicz equation \eqref{eqLichnerowicz} is bounded away from zero.
We give a slight improvement of \cite{IsenbergOMurchadha} and \cite[Theorem 1.7]{DahlGicquaudHumbert}
to the class of metrics with non-negative Yamabe invariant showing that the assumption
$\sigma \not\equiv 0$ is needed.

As in \cite{DahlGicquaudHumbert}, the manifold $M$ is still assumed to admit no conformal Killing
vector fields. Recall that the proof presented in \cite{DahlGicquaudHumbert} depends on a Sobolev quotient
for the operator $\bL_g$, i.e. whenever $M$ admits no non-zero conformal Killing vector fields, the following holds:
\begin{equation}\label{eqSobolevQuotient}
C_g = \inf_{0 \not\equiv V \in W^{1, 2}(M, T^*M)} \frac{{{{\left( {\int_M {|\bL_g V|_g^2 \dv_g} } \right)}^{1/2}}}}{{{{\left( {\int_M {|V|_g^N \dv_g} } \right)}^{1/N}}}} > 0
\end{equation}
The main result in this subsection is the following.

\begin{theorem}\label{thmClosedNonExistence}
Assume that $g \in W^{2,p}(M, S^2(M))$ has non-negative Yamabe invariant $\cY(g)$ and $(M,g)$ has no conformal Killing
vector fields. If $\sigma \equiv 0$ and $\tau \in W^{1, p}(M, \bR)$, there exists a positive constant $\mathcal C(g)$
independent of $\tau \in W^{1, p}(M, \bR)$ such that if \[\left\|\frac{d\tau}{\tau} \right\|_{L^n} < \mathcal C,\]
then there is no solution $(\phi, W)$ to the system \eqref{systemWithoutB} with $\phi > 0$.
\end{theorem}

Note that this allows (a priori) $\tau$ to have isolated non degenerate zeros. But, if $\tau$ changes sign, it
can be proven that $d\tau/\tau$ does not belong to any $L^p$ space for any $p \geqslant 1$. Hence, such a case
is out of reach from this theorem.

\begin{proof}
Let us first assume that the system \eqref{systemWithoutB} admits a solution $(\phi, W)$ with $\phi>0$
and $\sigma \equiv 0$. To prove the result, we denote by $\gbar$ the conformal metric $\psi^{N-2}g$ where
a positive function $\psi \in W^{2, p}(M, \bR)$ is chosen in such a way that $\scal_{\gbar}  \geqslant 0$.
Such a function $\psi$ exists since $\cY(g)\geqslant 0$. In terms of the metric $\gbar$, Equation
\eqref{systemWithoutB} becomes
\begin{equation}\label{hamiltonian-transformed}
 \begin{aligned}
-\frac{4(n-1)}{n-2} \Delta_{\gbar} (\psi ^{-1}\phi) + \scal_{\gbar} (\psi ^{-1}\phi) = &-\frac{n-1}{n} \tau^2 (\psi ^{-1}\phi)^{N-1}\\
&+ |\psi ^{-2}\bL_g W|_{\gbar} ^2 (\psi ^{-1}\phi)^{-N-1}.
 \end{aligned}
\end{equation}
Consequently, if we denote $\phibar \definedas \psi ^{-1}\phi$, multiply both sides of
\eqref{hamiltonian-transformed} by $\phibar^{N+1}$ and integrate both sides of the resulting
equation with respect to the conformal metric $\gbar$, we get
\begin{equation}\label{eq2}
\begin{aligned}
\frac{3n-2}{n-2} \int_M \big|d\phibar^{N/2+1}\big|^2_{\gbar} \dv_{\gbar} &+ \int_M \scal_{\gbar} \phibar^{N+2} \dv_{\gbar}
\\
& + \frac{n-1}{n} \int_M \tau^2 \phibar^{2N} \dv_{\gbar} = \int_M |\psi ^{-2}\bL_g W|_{\gbar}^2 \dv_{\gbar} .
\end{aligned}
\end{equation}
Under our conformal change $\gbar = \psi^{N-2}g$, there holds
\begin{equation}\label{eqConformalRules}
\begin{aligned}
\dv_{\gbar} &= \psi^{N} \dv_g ,\\
|{\psi ^{ - 2}}\bL_g W|_{\gbar}^2 &= {\psi ^{ - 2N}}|\bL_g W|_g^2.
\end{aligned}
\end{equation}
Therefore, in terms of the background metric $g$, \eqref{eq2} implies
\begin{equation}\label{eq3}
\frac{{n - 1}}{n}\int_M {{\tau ^2}} \psi^{-N} {\phi ^{2N}}\dv_g \leqslant \int_M {{\psi ^{ - N}}|\bL_g W|_g^2\dv_g} .
\end{equation}
Since $\psi \in W^{2, p}(M)$ is strictly positive, \eqref{eq3} immediately implies
\begin{equation}\label{eq4}
\int_M \tau ^2 \phi ^{2N} \dv_g \leqslant \frac{n}{n-1} \Big( \frac{\max \psi}{\min \psi} \Big)^N\int_M |\bL_g W|_g^2 \dv_g.
\end{equation}
We take the scalar product of the vector equation \eqref{eqVector} with $W$ and integrate over $M$ with respect to the background metric $g$ to get
\begin{equation}\label{eqIntegrationByParts}
- \frac{1}{2}\int_M | \bL_g W|_g^2 \dv_g = \frac{{n - 1}}{n}\int_M {{\phi ^N}\left\langle {d\tau ,W} \right\rangle } \dv_g.
\end{equation}
Using the H\"{o}lder inequality, we can estimate \eqref{eqIntegrationByParts} as follows:
\begin{equation}\label{grosCalcul}
\begin{aligned}
\frac{1}{2}\int_M |& \bL_g W|_g^2 \dv_g\\
  \leqslant & \frac{{n - 1}}{n} {\left( {\int_M {{\tau ^2}} {\phi ^{2N}}} \dv_g \right)^{1/2}} {\left( {\int_M {\Big| \frac{d\tau}{\tau} \Big|_g^n}} \dv_g \right)^{1/n}}{\left( {\int_M {|W|_g^{N}} } \dv_g \right)^{1/N}}\\
 \leqslant & \frac{{n - 1}}{n} \left(\frac{n}{n-1} \Big( \frac{\max \psi}{\min \psi} \Big)^N \int_M |\bL_g W|_g^2 \dv_g \right)^{1/2} \times \\
& \qquad \times \left( \int_M \Big| \frac{d\tau}{\tau} \Big|_g^n \dv_g \right)^{1/n} C_g^{-1} \left( \int_M |\bL_g W|_g^{2} \dv_g \right)^{1/2}\\
  \leqslant &\sqrt{\frac{{n - 1}}{n}} C_g^{-1} \Big(\frac{\max \psi}{\min \psi} \Big)^{N/2} \left( \int_M \Big| \frac{d\tau}{\tau} \Big|_g^n \dv_g \right)^{1/n} \int_M |\bL_g W|^{2} \dv_g.
\end{aligned}
\end{equation}
By setting
\[\mathcal C = \frac{1}{2} \sqrt{\frac{n}{n-1}} C_g \Big( \frac{\min \psi}{\max \psi} \Big)^{N/2},\]
one gets that
\[
\int_M \Big| \frac{d\tau}{\tau} \Big|_g^n \dv_g \geqslant \cC,
\]
unless
\[
\int_M |\bL_g W|_g^{2} \dv_g = 0.
\]
However, in the second case, we conclude from Inequality \eqref{eq3} that
\[
\int_M \tau^2 \psi^{-N} \phibar^{2N} \dv_g = 0.
\]
Hence $\phibar \equiv 0$ which contradicts the fact that $\phi > 0$. Thus, we have proved that if
$d\tau/\tau$ is small in the $L^n$--sense, the constraint equations \eqref{systemWithoutB} with vanishing
$\sigma$ admit no solution.
\end{proof}

Since our assumptions is weaker than those in \cite[Theorem 1.7]{DahlGicquaudHumbert}, for a price we pay,
the constant $\mathcal C$ appearing in Theorem \ref{thmClosedNonExistence} is smaller than the constant
appearing in \cite[Theorem 1.7]{DahlGicquaudHumbert}.

\section{The asymptotically Euclidean case}\label{secAE}
\renewcommand{\theclaim}{\arabic{claim}'}
\setcounter{claim}{0}

We now study the situation in the asymptotically Euclidean case. For relevant results
on Sobolev spaces on asymptotically Euclidean manifolds, we refer the reader to
\cite{Bartnik} or \cite{Maxwell}. See also the forthcoming article
\cite{DiltsGicquaudIsenberg}.

Let $(M^n, g)$ be a complete non-compact Riemannian manifold. We say that $(M, g)$ is
$W^{k,p}_\delta$--asymptotically Euclidean if there exist a compact set $K \subset M$, a real number $R > 0$,
and a diffeomorphism $\Psi: M \setminus K \to \bR^n \setminus B_R(0)$ such that, denoting $b$ the flat (background)
metric on $\bR^n$ and setting $e \definedas \Psi_* g - b$, we have
\[
\sum_{0 \leqslant i \leqslant k} \int_{\bR^n \setminus B_R} \big| \partial^{(i)} e\big|_b^p \left(1+|x|^2\right)^{-(\delta+ n/p -|i| )p/2} \dv_b(x) < \infty
\]
for some $k \geqslant 2$, $p > n$ and $\delta > 0$. Here, we denoted by $\partial^{(i)} e$ the $i$th order derivative (in the sense of distributions)
of $e$ and $\left| \partial^{(i)} e\right|_b$ its (pointwise) norm with respect to the Euclidean metric.

Given an asymptotically Euclidean manifold $(M, g)$ we denote by $r$ the pullback of the distance function from the origin in $\bR^n$:
$r = |\cdot| \circ \Psi$ and extend it to a positive continuous function on $K$. For any natural tensor bundle $E \to M$ and any section
$\xi \in \Gamma(E)$, we define the following weighted Sobolev norm:
\[
\left\|\xi\right\|_{W^{s, q}_\gamma(M, E)} \definedas \bigg(\sum_{0 \leqslant i \leqslant s} \int_M \big|\nabla^{(i)} \xi \big|_g^q \big(1+r^2\big)^{-(\gamma+ n/p -|i| )q/2} \dv_g\bigg)^{1/q},
\]
and the associated Sobolev space
\[
W^{s, q}_\delta (M, E) \definedas \left\{\xi \in W^{s, q}_{\rm loc}, \left\|\xi\right\|_{W^{s, q}_\delta(M, E)} < \infty\right\}.
\]
We also recall that the Yamabe invariant for an asymptotically Euclidean manifold $(M, g)$ is given by
\eqref{eqDefYamabe} even if the solution to the Yamabe problem in this case does not belong to
$W^{1, 2}$ since it tends to some positive constant at infinity.

We prove the following theorem.

\begin{theorem}\label{thmAE}
Let $(M, g)$ be a $W^{2, p}_\delta$--asymptotically Euclidean manifold for some $p > n$ and some
$\delta \in (2-n, 0)$. Assume that the Yamabe invariant $\cY(g)$ of the manifold $(M, g)$ is positive.
Then given any $\tau \in W^{1, p}_\delta(M, \bR)$, $\sigmatil \in W^{1, p}_\delta(M, S^2(M))$,
$\sigmatil \not\equiv 0$, and $\phitil_\infty \in \bR_+^*$, there exists $\eta_0 > 0$ such that for any
$\eta \in (0, \eta_0)$ there exists at least one solution to the system \eqref{eqLichnerowicz}--\eqref{eqVector}
with $\sigma = \eta \sigma_0$ and $(\phi - \eta^{2/(N-2)} \phitil_\infty, W) \in W^{2, p}_\delta(M, \bR) \times W^{2, p}_\delta(M, T^*M)$. 
\end{theorem}

Note that the condition $\phi - \eta \phitil_\infty \in W^{2, p}_\delta(M, \bR)$ immediately implies that
$\phi \to \eta \phitil_\infty$ at infinity. The proof of this theorem mimics that of Theorem
\ref{thmClosed} replacing the $W^{k, p}$--spaces by the $W^{k, p}_\delta$ ones. We only give the
analogs of each of the four claims and a proof of the significantly different steps.

\begin{claim}\label{clLichneroLimitAE}
There exists a unique solution $\phitil_0$ to the equation \eqref{eqLichnerowicz0} such that
$\phitil_0 - \phitil_\infty \in W^{2, p}_\delta(M, \bR)$.
\end{claim}

\begin{proof}
To simplify the proof, we assume that the manifold $(M, g)$ has zero scalar curvature. This assumption
is harmless since it is known that any asymptotically Euclidean metric $g$ with positive Yamabe invariant
$\cY(g)$ is conformally related to a metric $\gbar = \psi^{N-2} g$ with zero scalar curvature with
$\psi - 1 \in W^{2, p}_\delta(M, \bR)$ (for instance, see \cite[Proposition 3]{Maxwell}). Hence, one can
proceed as in the proof of Claim \ref{clLichneroLimit}, working with metric $\gbar$ and replacing
$\left|\sigmatil\right|_g^2$ by $\left|\psi^{-2}\sigmatil\right|_g^2$.

To prove the existence part, we first decompose $\phitil = \phitil_\infty + \vtil$ and wish to look for
$\vtil$ in $W^{2, p}_\delta(M, \bR)$ solving the following PDE:
\begin{equation}\label{eqVtil}
- \frac{4(n-1)}{n-2} \Delta_g \vtil = \frac{\left|\sigmatil\right|_g^2}{\left(\phitil_\infty + \vtil\right)^{N+1}}.
\end{equation}
Note that $\vtil_- \equiv 0$ is always a subsolution to \eqref{eqVtil}. To construct a super-solution to
\eqref{eqVtil}, let $\vtil_+ \in W^{2, p}_\delta(M, \bR)$ denote the solution to the following Poisson equation:
\[
- \frac{4(n-1)}{n-2} \Delta_g \vtil_+ = \frac{\left|\sigmatil\right|_g^2}{\left(\phitil_\infty\right)^{N+1}}.
\] 
From the strong maximum principle it follows that $\vtil_+ > 0$. As a consequence, there holds
\[
- \frac{4(n-1)}{n-2} \Delta_g \vtil_+ \geqslant \frac{\left|\sigmatil\right|_g^2}{\left(\phitil_\infty + \vtil_+\right)^{N+1}},
\]
this is to say that $\vtil_+$ is a supersolution to \eqref{eqVtil}. The standard sub- and super-solutions
method applies giving rise to the existence of a solution $\phitil_0$ solving \eqref{eqLichnerowicz0} and
satisfying $\phitil_0 - \phitil_\infty \in W^{2, p}_\delta(M, \bR)$.

The proof of the uniqueness property is then entirely similar to the compact case, therefore we omit it.
\end{proof}

\begin{claim}\label{clDeformedAE}
There exists $\varepsilon > 0$ such that the system \eqref{DeformedsystemWithoutB} admits a solution
$(\phitil_\mu, \Wtil_\mu)$ such that $\phitil_\mu - \phitil_\infty \in W^{2, p}_\delta(M, \bR)$ and
$\Wtil_\mu \in W^{2, p}_\delta(M, T^*M)$ for all $\mu \in [0, \varepsilon)$.
\end{claim}

\begin{proof}
The proof of Claim \ref{clDeformed} translates mutatis mutandis, the only difference being that we need
to work on the affine space $(\phitil_\infty, 0) + W^{2, p}_\delta(M, \bR) \times W^{2, p}_\delta(M, T^*M)$.
The relevant properties for the operator $\DeltaL$ on asymptotically Euclidean manifolds can be found
in \cite[Theorem 5.4]{Maxwell}.
\end{proof}

\begin{claim}\label{clScalingAE}
Set
\[
\left\lbrace
\begin{aligned}
\phi_\mu & = \mu^{\frac{2}{N-2}} \phitil_\mu,\\
W_\mu  & = \mu^{\frac{N+2}{N-2}} \Wtil_\mu,\\
\sigma_\mu & = \mu^{\frac{N+2}{N-2}} \sigmatil.
\end{aligned}
\right.
\]
If $(\phitil_\mu, \Wtil_\mu)$ solves \eqref{DeformedsystemWithoutB} with $\phitil_\mu \to \phitil_\infty$
at infinity, then $(\phi_\mu, W_\mu)$ solves \eqref{systemWithoutB} with $\sigma = \sigma_\mu$ and
$\phi_\mu \to \mu^{2/(N-2)} \phitil_\infty$ at infinity.
\end{claim}

\section{The compact with boundary case}\label{secCWB}

\subsection{Boundary conditions}\label{secBoundaryConditions}
A natural issue in the study of the Einstein constraint equations is the construction of initial
data modeling black holes. While the definition of a black hole requires knowledge of the whole
solution $(\cM, \gbold)$ of the Einstein equations, it is natural to construct initial data containing
apparent horizons. For an overview, we refer the reader to \cite{ChruscielGallowayPollack}. A natural
way to construct such solutions is to excise the inside of the apparent horizon and thus construct
solutions to the constraint equations on the outside. As a consequence, we fix a manifold $M$
with boundary $\partial M$, solve the constraint equations on $M$ in such a way that $\partial M$
becomes an apparent horizon.

The first articles where such solutions to the constraint equations were constructed dealt with the
constant mean curvature case, see e.g. \cite{Maxwell, Gicquaud}. Very recently, people have turned
their attention to compact manifolds with boundary with a varying $\tau$, see for example
\cite{Dilts, HolstMeierTsogtgerel}. 

To go further, let us roughly reformulate this problem. For detailed explanation and calculations, we refer
the reader to \cite{Dilts, HolstMeierTsogtgerel, GicquaudHumbertNgo}. Let $\nuhat$ be the
(spacelike) unit normal vector field to $\partial M$ in $M$ pointing towards the outside of $M$ (hence to
the ``inside'' of the apparent horizon) and let $\nhat$ be the future directed unit normal spacetime vector
field to $M$. Then, by means of apparent horizon boundaries, in addition to the constraint equations
\eqref{eqConstraintEquations}, we further require
\begin{equation}\label{eqConstraintEquationsOnBoundary}
\left\lbrace
\begin{aligned}
\widehat{\Theta}_- & \leqslant 0,\\
\widehat{\Theta}_+ & = 0,
\end{aligned}
\right.
\end{equation}
where $\widehat{\Theta}_\pm$, known as the null expansion with respect to the null normal
$\ell_\pm \definedas \widehat n \mp \nuhat$, are given as follows:
\[
\widehat{\Theta}_\pm = \tr_\ghat \Khat - \Khat(\nuhat, \nuhat) \mp H_\ghat
\]
where $H_\ghat$ is the (unnormalized) mean curvature of $\partial M$ in $M$ evaluated with respect to
$\nuhat$, that is to say
\[ H_\ghat = \ghat^{ij} \nablahat_i \nuhat_j,\]
where we denote by $\nablahat$ the Levi-Civita connection for the metric $\ghat$.
Since we require $\widehat{\Theta}_+ \equiv 0$ on $\partial M$, the conditions can be rewritten as
\[
\left\lbrace
\begin{aligned}
\tr_\ghat \Khat - \Khat(\nuhat, \nuhat) & = \frac{\widehat{\Theta}_+ + \widehat{\Theta}_-}{2} = \frac{\widehat{\Theta}_-}{2},\\
H_\ghat & = \frac{\widehat{\Theta}_- - \widehat{\Theta}_+}{2} = \frac{\widehat{\Theta}_-}{2}.
\end{aligned}
\right.
\]
On the other hand, recalling that $\ghat = \phi^{N-2} g$, one has the following formula relating:
$H_\ghat$ and $H_g$:
\[
\frac{2(n-1)}{n-2}\partial_\nu \phi + H_g \phi = H_\ghat \phi^{N/2},
\]
where $\nu = \phi^{N/2-1} \nuhat$ is the unit vector field normal to $\Sigma$ calculated with respect
to the metric $g$. Hence, we get the following condition for $\phi$:
\begin{equation}\label{eqBoundaryOfPhiInTermOfTheta}
\frac{2(n-1)}{n-2}\partial_\nu \phi + H_g \phi = \frac{\widehat\Theta_-}{2} \phi^{N/2}.
\end{equation}
Next, thanks to $\tr_\ghat \Khat = \tau$ and the fact that
\[
\Khat(\nuhat, \nuhat) = \frac \tau n + (\sigma + \bL_g W)(\nu, \nu) \phi^{-N},
\] 
we obtain the following identity:
\begin{equation}\label{eqBoundaryNuNu}
\frac{\widehat{\Theta}_-}{2} = \frac {n-1}{n} \tau - (\sigma + \bL_g W)(\nu, \nu) \phi^{-N}.
\end{equation}
Contrary to \eqref{eqBoundaryOfPhiInTermOfTheta}, this does not give a boundary condition that complements Equation
\eqref{eqVector}. In this context, it is natural to prescribe $(\sigma + \bL_g W)(\nu, \cdot)$ as follows:
\begin{equation}\label{eqBoundaryW}
(\sigma + \bL_g W)(\nu, \cdot) = \bigg(\frac {n-1}{n} \tau - \frac{\widehat{\Theta}_-}{2}\bigg) \phi^N \nu^\flat + \xi
\end{equation}
where $\xi$ is a $1$--form on $\partial M$ which we extend to the restriction of $TM$ to $\partial M$ by
setting $\xi(\nu) = 0$ so that Condition \eqref{eqBoundaryNuNu} is satisfied. Also in \eqref{eqBoundaryW},
we use $\nu^\flat$ to denote the $1$--form dual to the normal vector field $\nu$ which is given by
$\nu^\flat (X) = g(\nu, X)$ for any vector field $X$ on $\partial M$. Having all discussion
above, we are now in a position to write down the following system of PDEs:
\begin{equation}\label{systemWithB}
\left\lbrace
\begin{aligned}
-\frac{4(n-1)}{n-2} \Delta_g \phi + \scal_g \phi & = - \frac{n-1}{n} \tau^2 \phi^{N-1} + \big| \sigma + \bL_g W \big|_g^2 \phi^{-N-1},\\
\DeltaL W & = \frac{n-1}{n} \phi^N d\tau,\\
\frac{2(n-1)}{n-2}\partial_\nu \phi + H_g \phi &= \frac{\widehat\Theta_-}{2} \phi^{N/2},\\
(\sigma + \mathbb{L}W)(\nu, \cdot) &= \bigg({\frac{{n - 1}}{n}\tau - \frac{\widehat\Theta_-}{2}} \bigg){\phi^N} \nu^\flat + \xi,
\end{aligned}
\right.
\end{equation}
where the given data are now $(M, g)$ a compact Riemannian manifold with boundary $\partial M$,
$\tau$ a function on $M$, $\sigma$ a TT-tensor, $\widehat{\Theta}_-$ a nonpositive function on
$\Sigma = \partial M$ and $\xi \in \Gamma(\partial M, T^*M)$ a 1--form.

In the presence of the boundary $\partial M$, instead of using the sign of $\cY(g)$,
we use the sign of the Yamabe invariant $\cY(g, \partial M)$ introduced by Escobar \cite{Escobar}:
\[
\cY(g, \partial M) \definedas \inf_{0 \not\equiv \phi \in W^{1, 2}(M, \bR)} \frac{\int_M \big( \frac{4(n-1)}{n-2} |d\phi|^2_g + \scal_g \phi^2 \big) \dv_g + \int_{\partial M} H_g \phi^2 \ds_g}
  {\left(\int_M \phi^N \dv_g\right)^{N/2}}.
\]
We also comment on the York splitting on compact manifolds with boundary. While on closed manifolds
we have that the set of (say) $W^{1, 2}$--TT-tensors is $L^2$--orthogonal to the set $\{\bL_g W, W \in W^{2,2}(M, T^*M)\}$,
this is no longer true for compact manifolds with boundary. Indeed, let $\sigma$ be a TT-tensor and
$W$ be an arbitrary $1$--form, then if we denote by $W^\sharp$ the vector field dual to the $1$--form $W$, then by a direct calculation together with the Stokes theorem, we have
\[
\begin{aligned}
\int_M \left\<\sigma, \bL_g W\right\> \dv_g
 & = 2 \int_M \left\<\sigma, \nabla W\right\> \dv_g\\
 & = 2 \int_M \divg \left(\sigma(W^\sharp, \cdot)\right)\dv_g - 2 \int_M (\divg \sigma)(W^\sharp) \dv_g\\
 & = 2 \int_{\partial M} \sigma(W^\sharp, \nu) \ds_g,
\end{aligned}
\]
where $\tr_g \sigma=0$ and $\divg_g \sigma=0$ were also used to obtain the first and last lines respectively. Since the restriction of $W$ to $\partial M$ can be arbitrary, $\sigma$ belongs to the orthogonal of
the set of $\bL_g W$'s if and only if we also impose that $\sigma(\nu, \cdot) \equiv 0$ on $\partial M$.
We will make this assumption from now on.

\subsection{Existence result}

\renewcommand{\theclaim}{\arabic{claim}''}
\setcounter{claim}{-1}

The main result of this subsection is the following.

\begin{theorem}\label{thmCompactWithB} 
Let $M$ be a compact manifold with boundary. Given $p > n$, let $g \in W^{2, p}(M, S^2(M))$, $\tau \in W^{1, p}(M, \bR)$,
and $\sigmatil \in W^{1, p}(M, S^2(M))$, $\widehat{\Theta}_- \in W^{1-1/p, p}(\partial M, \bR)$,
$\xitil \in W^{1-1/p, p}(\partial M, T^*M)$ be given data, where $\sigmatil$ is a TT-tensor such that
$\sigmatil(\nu, \cdot) \equiv 0$ on $\partial M$. Assume that the Escobar invariant $\cY(g, \partial M)$
is strictly positive, that $g$ has no conformal Killing vector fields and either $\sigmatil \not\equiv 0$ or
$\xitil \not\equiv 0$. There exists $\eta_0 > 0$ such that for any $\eta \in (0, \eta_0)$ there exists at least
one solution $(\phi, W) \in W^{2, p}(M, \bR) \times W^{2, p} (M, T^*M)$ to the system \eqref{systemWithB} with
$\sigma = \eta \sigmatil$ and $\xi = \eta \xitil$.
\end{theorem}

We initiate the proof of Theorem \ref{thmCompactWithB} by proving that the right hand side of the analog of Equation
\eqref{eqLichnerowicz0} (see Equation \eqref{System1} below) is actually non-zero.

\begin{claim}\label{clPhiNonZero}
Let $\Wtil_0 \in W^{2, p}(M, T^*M)$ be the unique solution of
\[
\left\lbrace
\begin{aligned}
\DeltaL \Wtil_0 & = 0,\\
\bL_g(\nu, \cdot) & = \xitil.
\end{aligned}
\right.
\]
Then under the assumptions of Theorem \ref{thmCompactWithB}, we have
\[
\big|\sigmatil + \bL_g \Wtil_0\big|_g^2 \not\equiv 0.
\]
\end{claim}

\begin{proof}
The existence, the uniqueness, and the regularity of $\Wtil_0$ are proved in \cite[Theorem 4.5]{HolstMeierTsogtgerel}.
See also \cite[Proposition 5.1]{Maxwell} and \cite[Theorem 8.6]{Gicquaud} for earlier references. From the remark
at the end of Subsection \ref{secBoundaryConditions}, we have
\[
\int_M \big|\sigmatil + \bL_g \Wtil_0\big|_g^2 \dv_g = \int_M \left|\sigmatil\right|_g^2 \dv_g + \int_M \big|\bL_g \Wtil_0\big|_g^2 \dv_g.
\]
Hence if $\sigmatil \not\equiv 0$, the claim follows. Otherwise if $\xitil \not\equiv 0$, $\Wtil_0$ is a non-trivial
element of $W^{2, p}(M, T^*M)$. Since $(M, g)$ has no non-zero conformal Killing vector field, it follows that
\[\int_M \big|\bL_g \Wtil_0\big|_g^2 \dv_g > 0,\] 
which proves the claim.
\end{proof}

\begin{claim}\label{clLichneroLimitWithB}
Under the assumptions of Theorem \ref{thmCompactWithB}, there exists a unique solution $\phitil_0 \in W^{2, p}(M, \bR)$
to the following system:
\begin{equation}\label{System1}
\left\lbrace
 \begin{aligned}
-\frac{4(n-1)}{n-2} \Delta_g \phitil_0+ \scal_g \phitil_0 &= \big|\sigmatil + \bL_g \Wtil_0\big|_g^2 \phitil_0^{-N-1},\\
\frac{2(n-1)}{n-2}\partial_\nu \phitil_0 +H_g \phitil_0 &=0.
 \end{aligned}
\right.
\end{equation}
\end{claim}

\begin{proof}
The proof of this claim is similar to the proof of Claim \ref{clLichneroLimit}. From the work of Escobar
\cite[Lemma 1.1]{Escobar}, there exists a conformal factor $\psi \in W^{2, p}(M, \bR)$ such that the metric $\gbar = \psi^{N-2} g$
has $\scal_{\gbar} > 0$ and the mean curvature of the boundary $\partial M$ vanishes identically: $H_{\gbar} \equiv 0$
\footnote{As pointed out by one of the referees, \cite[Lemma 1.1]{Escobar} is only stated for smooth metrics. However, the proof
works for $W^{2, p}$-metrics without any change.}.
The equation for $\phibar_0 \definedas \psi^{-1} \phitil_0$ reads
\begin{equation}\label{System1p}
\left\lbrace
 \begin{aligned}
-\frac{4(n-1)}{n-2} \Delta_g \phibar_0+ \scal_g \phibar_0 &= \big|\psi^{-2}\big(\sigmatil + \bL_g \Wtil_0\big)\big|_{\gbar}^2 \phibar_0^{-N-1},\\
 \partial_{\nubar} \phibar_0 & = 0,
 \end{aligned}
\right.
\end{equation}
where $\nubar = \psi^{1-N/2} \nu$ is the unit normal to $\partial M$ for the metric $\gbar$. There exists
a unique function $\ubar \in W^{2, p}(M, \bR)$ solving
\begin{equation}\label{System1u}
\left\lbrace
 \begin{aligned}
-\frac{4(n-1)}{n-2} \Delta_g \ubar + \scal_g \ubar &= \big|\psi^{-2}\big(\sigmatil + \bL_g \Wtil_0\big)\big|_{\gbar}^2,\\
 \partial_{\nubar} \ubar_0 & = 0.
 \end{aligned}
\right.
\end{equation}
Further, the function $u$ is positive. By setting
\[
\phibar_- = \left(\max \ubar\right)^{-\frac{N+1}{N+2}} \ubar
\]
and
\[
\phibar_+ = \left(\min \ubar\right)^{-\frac{N+1}{N+2}} \ubar,
\]
one readily checks that $\phibar_+$ and $\phibar_-$ are super- and sub-solutions for \eqref{System1p}. Hence, by the
sub- and super-solution method, we conclude that there exists a solution $\phibar_0$ to \eqref{System1p}. The function
$\phitil_0 \definedas \psi \phibar_0$ is then a solution to \eqref{System1}. The proof of uniqueness is a rephrasing
of that in Claim \ref{clLichneroLimit} with a Neumann boundary condition.
\end{proof}

Similar to \eqref{DeformedsystemWithoutB} for the closed case, in view of \eqref{systemWithB} we now introduce the
following $\mu$--deformed system for the compact with boundary case:
\begin{equation}\label{DeformedSystem}
\left\lbrace
 \begin{aligned}
-\frac{4(n-1)}{n-2} \Delta_g \phitil + \scal_g \phitil &= -\frac{n-1}{n} \tau^2 \mu^2 \phitil^{N-1} + |\sigmatil + \bL_g \Wtil|_g^2 \phitil^{-N-1},\\
\DeltaL \Wtil & = \frac{n-1}{n} \phitil^{N} \mu d\tau,\\
\frac{2(n-1)}{n-2}\partial_\nu \phitil + H_g \phitil &= \frac{\widehat{\Theta}_-}{2} \mu \phitil^{N/2},\\
\bL_g \Wtil(\nu, \cdot) &= \mu \left( \frac{n-1}{n} \tau - \frac{\widehat{\Theta}_-}{2} \right) \phitil^N \nu^\flat + \xitil.
 \end{aligned}
\right.
\end{equation}

This system is obtained from \eqref{systemWithB} by replacing $\tau$ by $\mu \tau$ and $\widehat{\Theta}_-$ by $\mu \widehat{\Theta}_-$.

\begin{claim}\label{clDeformedWithB}
There exists $\varepsilon>0$ such that \eqref{DeformedSystem} admits a solution $(\phitil_\mu, \Wtil_\mu)$ for all $\mu \in [0, \varepsilon)$.
\end{claim}

\begin{proof}
We define the following operator:
\[\begin{array}{*{20}{c}}
 {F:}&{\bR \times W_+^{2,p}(M,\bR) \times W^{2,p}(M, T^*M)} \\ 
 { }& \downarrow \\
 {}&{L^p(M, \bR) \times W^{1-\frac 1p, p}(\partial M, \bR) \times L^p(M, T^*M) \times W^{1-\frac 1p, p}(\partial M, T^*M)} 
\end{array}\]
given by
\[{\setstretch{2.20}
F(\mu, \phitil, \Wtil)=
\begin{pmatrix} 
 {-\frac{4(n-1)}{n-2} \Delta_g \phitil + \scal_g \phitil + \frac{n-1}{n} \tau^2 \mu^2 \phitil^{N-1} - |\sigma+ \bL_g \Wtil|_g^2 \phitil^{-N-1} }\\ 
 {\frac{2(n-1)}{n-2}\partial_\nu \phitil +H_g \phitil - \frac{\widehat{\Theta}_-}{2} \mu \phitil^{N/2}} \\ 
  {\DeltaL \Wtil - \frac{n-1}{n} \phitil^{N} \mu d\tau} \\ 
  {\bL_g\Wtil(\nu, \cdot) - \mu \left( \frac{n-1}{n} \tau - \frac{\widehat{\Theta}_-}{2} \right) \phitil^N \nu^\flat} - \xitil
\end{pmatrix} .}
\]
It is not hard to see that the mapping $F$ is of class $C^1$ and
\[
F(0, \phitil_0, \Wtil_0)= \begin{pmatrix} 0\\0\\0\\0 \end{pmatrix},
\]
where $\phitil_0$ and $\Wtil_0$ are given in Claims \ref{clPhiNonZero} and \ref{clLichneroLimitWithB}.
Again, all we need to do is to prove that the derivative of $F$ with respect to $(\phitil, \Wtil)$
is an isomorphism at $(0, \phitil_0, \Wtil_0)$. To do so, we need to study the following mapping:
\[
\begin{array}{*{20}{c}}
 {\mathcal DF_{(0, \phitil_0, \Wtil_0)}:} & {W^{2,p}(M,\bR) \times W^{2,p}(M, T^*M)} \\ 
 { }& \downarrow \\
 {}&{L^p(M, \bR) \times W^{1-\frac 1p, p}(\partial M, \bR) \times L^p(M, T^*M) \times W^{1-\frac 1p, p}(\partial M, T^*M).} 
\end{array}
\]
A direct computation shows that this derivative is given by
\[\begin{split}
\mathcal D &F_{(0, \phitil_0, \Wtil_0)} (\thetatil, \Ztil)\\
&= {\setstretch{2.00}\left( \begin{array}{c:c} {-\frac{4(n-1)}{n-2} \Delta_g + \scal_g +(N+1) |\sigmatil + \bL_g \Wtil_0|_g^2 \phitil_0^{-N-2} } & -2 \big\langle\sigmatil + \bL_g \Wtil_0, \bL_g \cdot\big\rangle\\ 
\hdashline \frac{2(n-1)}{n-2}\partial_\nu +H_g & 0 \\ 
\hdashline 0 & \DeltaL \\ 
\hdashline 0 & \bL_g \cdot (\nu, \cdot) \end{array} \right)}
\begin{pmatrix} \thetatil \\ \Ztil \end{pmatrix}.
\end{split}\]
Clearly, $\mathcal DF_{(0, \phitil_0, \Wtil_0)}$ is continuous. To prove that $\mathcal DF_{(0, \phitil_0, \Wtil_0)} $ is
invertible, we observe that $\mathcal DF_{(0, \phitil_0, \Wtil_0)}$ is block upper-triangular, where the diagonal blocks
are
\[{\setstretch{1.80}
\begin{pmatrix}
{-\frac{4(n-1)}{n-2} \Delta_g + \scal_g +(N+1) |\sigma|_g^2 \phitil_0^{-N-2} } \\ 
\frac{2(n-1)}{n-2}\partial_\nu +H_g 
\end{pmatrix}}
\qquad\text{and}\qquad {\setstretch{1.80}\begin{pmatrix} \DeltaL \\ \bL_g \cdot (\nu, \cdot) \end{pmatrix}} \]
which are invertible. Hence, the derivative $\mathcal DF_{(0, \phitil_0, \Wtil_0)}$ is an isomorphism at $(0, \phitil_0, \Wtil_0)$ as claimed.
\end{proof}

\begin{claim}\label{clScalingWithB}
Set
\[
\left\lbrace
\begin{aligned}
\phi_\mu & = \mu^{\frac{2}{N-2}} \phitil_\mu,\\
W_\mu  & = \mu^{\frac{N+2}{N-2}} \Wtil_\mu.
\end{aligned}
\right.
\]
If $(\phitil_\mu, \Wtil_\mu)$ solves \eqref{DeformedSystem}, then $(\phi_\mu, W_\mu)$ solves \eqref{systemWithB} with
$\sigma = \sigma_\mu \definedas \mu^{\frac{N+2}{N-2}} \sigmatil$ and $\xi = \xi_\mu \definedas \mu^{\frac{N+2}{N-2}} \xitil$.
\end{claim}

Finally, the proof of Theorem \ref{thmCompactWithB} follows by setting $\eta_0 = \varepsilon^{\frac{N+2}{N-2}}$, where
$\varepsilon$ is the constant appearing in Claim \ref{clDeformedWithB}.

\begin{remark}
It is tempting to prove an analog of the non-existence result for the case of a compact manifold with boundary as in Theorem \ref{thmClosedNonExistence}. The natural assumptions in this theorem would then be $\sigma \equiv 0$, $\xi \equiv 0$ and $\cY(g, \partial M) > 0$. The proof is however not just an extension of that of Theorem \ref{thmClosedNonExistence}, it relies on techniques developed in \cite{GicquaudHumbertNgo} so we choose to defer it to that article.
\end{remark}

\section*{Acknowledgments} 

The authors are grateful to Michael Holst, Emmanuel Humbert, Jim Isenberg, David Maxwell and Daniel Pollack for useful discussions
and advises. We would also like to thank warmly anonymous referees and The Cang Nguyen for their very careful proofreadings of this article.
Part of this work was done when the first author was attending the program \emph{Mathematical General Relativity}
at the Mathematical Sciences Research Institute in Berkeley, California, during the Fall of 2013. The second author
also wants to acknowledge the support of the R\'{e}gion Centre through the convention n$^{\rm o}$ 00078780 during the period 2012--2014.




\providecommand{\bysame}{\leavevmode\hbox to3em{\hrulefill}\thinspace}
\providecommand{\MR}{\relax\ifhmode\unskip\space\fi MR }
\providecommand{\MRhref}[2]{%
  \href{http://www.ams.org/mathscinet-getitem?mr=#1}{#2}
}
\providecommand{\href}[2]{#2}

\end{document}